\documentclass[twoside,leqno]{article}

\usepackage[letterpaper]{geometry}

\usepackage{ltexpprt}
\usepackage{amsmath}
\usepackage{amssymb}
\usepackage{url}
\usepackage[usenames,dvipsnames]{xcolor}
\usepackage{graphicx}
\usepackage[colorlinks=true, allcolors=blue]{hyperref}

\usepackage{latexsym}

\newcommand{\on}[2]{\mathop{\null#2}\limits^{#1}}
\newcommand{\bvec}[1]{\on{\,{}_\leftarrow}{{#1}}}

\newcommand{\RR}{\mathbb{R}}

\newcommand{\MCP}{\mbox{\sl MCP}}
\newcommand{\MFC}{\mbox{\sl MFC}}
\newcommand{\MIC}{\mbox{\sl MIC}}
\newcommand{\MPP}{\mbox{\sl MPP}}
\newcommand{\APSP}{\mbox{\sl APSP}}

\newcommand{\EQ}{\;=\;}

\begin{document}

% \title{Minimum cost paths with varying refueling/recharging costs}
\title{Minimum-cost paths for electric cars}

%\author{}

%\author{Dani Dorfman\thanks{Blavatnik School of Computer Science, Tel Aviv University, Israel. Email: {\tt dani.i.dorfman@gmail.com}, {\tt \{haimk,zwick\}@tau.ac.il.} Work of Uri Zwick partially supported by grant 2854/20 of the Israeli Science Foundation. Work of Haim Kaplan partially supported by ISF grant 1595/19 and the Blavatnik family foundation.} \and Haim Kaplan${}^*$ \and Robert E.\ Tarjan\thanks{Department of Computer Science, Princeton University. Email: {\tt ret@princeton.edu.}} \and Mikkel Thorup\thanks{BARC, University of Copenhagen, Denmark. Email: {\tt mikkel2thorup@gmail.com}} \and Uri Zwick${}^*$}
% \date{May 22, 2021}

\author{%
  \begin{tabular}{ c c c c c}
       Dani Dorfman\thanks{Blavatnik School of Computer Science, Tel Aviv University, Israel. Email: {\tt dani.i.dorfman@gmail.com}, {\tt \{haimk,zwick\}@tau.ac.il.} Work of Uri Zwick partially supported by grant 2854/20 of the Israeli Science Foundation. Work of Haim Kaplan partially supported by ISF grant 1595/19 and the Blavatnik family foundation. Work of Dani Dorfman partially supported by the The Israeli Smart Transportation Research Center (ISTRC).}\;   &    Haim Kaplan${}^*$\;   &    Robert E.\ Tarjan\thanks{Department of Computer Science, Princeton University. Research partially supported by a gift from Microsoft. Email: {\tt ret@princeton.edu.}}\;  
       & Mikkel Thorup\thanks{BARC, University of Copenhagen, Denmark. Research supported by the VILLUM Foundation grant no.\ 16582. Email: {\tt mikkel2thorup@gmail.com}}\;    &   Uri Zwick${}^*$        
  \end{tabular}
}

\date{}
\maketitle

% Default Copyright Statement
\fancyfoot[R]{\scriptsize{Copyright \textcopyright\ 2024 by SIAM\\
Unauthorized reproduction of this article is prohibited}}

% \vspace*{-12pt}
\begin{abstract}
    An electric car equipped with a battery of a finite capacity travels on a road network with an infrastructure of charging stations. Each charging station has a possibly different cost per unit of energy. Traversing a given road segment requires a specified amount of energy that may be positive, zero or negative.  The car can only traverse a road segment if it has enough charge to do so (the charge cannot drop below zero), and it cannot charge its battery beyond its capacity. %More precisely, the charge in the battery can never drop below zero and it can never exceed its capacity. 
    To travel from one point to another the car needs to choose a \emph{travel plan} consisting of a path in the network and a recharging schedule that specifies how much energy to charge at each charging station on the path, making sure of having enough energy to reach the next charging station or the destination. The cost of the plan is the total charging cost along the chosen path. We reduce the problem of computing  plans between every two junctions of the network to two problems: Finding optimal energetic paths when no charging is allowed and finding standard shortest paths. When there are no negative cycles in the network, we obtain an $O(n^3)$-time algorithm for computing all-pairs  travel plans, where~$n$ is the number of junctions in the network. We obtain slightly faster algorithms under some further assumptions. We also consider the case in which a bound is placed on the number of rechargings allowed.
\end{abstract}

% \vspace*{-45pt}
% \begin{abstract}
% \setlength{\parindent}{0pt}
% \setlength{\parskip}{3pt plus 2pt}\noindent%
% \end{abstract}

\section{Introduction}\label{S:intro}

An electric car with a battery that can store up to~$B$ units of energy travels in a road network with an infrastructure of charging stations. The road network is represented by a weighted directed graph $G=(V,A,c,r)$, where $V$ is the set of vertices (junctions), $A\subseteq V\times V$ is the set of arcs (road segments), $c:A\to\RR$ and $r:V\to \RR^+\cup\{\infty\}$. We always let $n=|V|$ and $m=|A|$. For $a\in A$, $c(a)$ is the amount of energy needed to traverse arc~$a$. This amount can be positive, zero or negative. A negative cost may indicate a downhill segment that can be used to charge the battery. The battery can be charged, at a certain cost, at some vertices of the network. For every $v\in V$, $r(v)\ge 0$ is the price for one unit of energy at~$v$. If $r(v)<\infty$, we say that~$v$ is a \emph{charging station}. If $r(v)=\infty$ then the battery cannot be charged at~$v$. The car can traverse an arc $(u,v)\in A$ only if it is in~$u$ and the charge~$b\ge 0$ in its battery satisfies $b\ge c(u,v)$. It then reaches~$v$ with a charge of $\min\{b-c(u,v),B\}$ in its battery. The charge~$b$ of the battery always satisfies $0\le b\le B$, i.e., the charge in the battery cannot drop below zero and cannot exceed the capacity~$B$.

If $c:A\to \RR^+$, i.e., all arc costs are non-negative, then the graph $G=(V,A,c,r)$ models a traditional road network with an infrastructure of gas stations in which fuel-based cars can travel.

To travel from~$s$ to~$t$ the car needs to choose a \emph{travel plan} composed of a directed path~$P$ from~$s$ to~$t$, and a \emph{recharging schedule} that specifies how much energy to charge at every charging station on~$P$, making sure that the car always has enough energy to continue its journey. We assume that the car starts at~$s$ with an empty battery. The \emph{cost} of the plan is the total charging cost along the selected path. We let $\rho_B(s,t)$ be the minimum cost of a travel plan from~$s$ to~$t$, starting at~$s$ with an empty battery, where~$B$ is the capacity of the battery. We also let $\rho_{B,\Delta}(s,t)$ be the minimum cost of a plan that uses at most~$\Delta$ charging stops.

% The problem of finding minimum cost plans is a `non-standard' variant of the shortest paths problem. One interesting aspect of it, which makes it non-trivial, is that it does not have the subpath optimality 

Khuller, Malekian and Mestre~\cite{khuller2011fill} considered the problem of finding  plans in the context of conventional fuel-based cars, i.e., $c:A\to \RR^+$, and obtained an $O(n^3\min\{\Delta\log n,\Delta^2\})$-time algorithm for the all-pairs version, where $\Delta\le n$ is a bound on the number of chargings allowed. When no bound is placed on the number of times the battery can be charged, the running time of their algorithm is $O(n^4\log n)$. They also obtained an $O(\min\{n^3,\Delta n^2\log n\})$-time algorithm for the single-target version of the problem.

Dorfman et al.~\cite{DKTZ23} recently considered a related problem of finding optimal paths for electric cars when no rechargings along the way are allowed, but when arc costs are also allowed to be negative. When there are no negative arc costs, the problem is equivalent to the standard shortest paths problem. When there are negative arc costs the problem becomes a strict extension of the standard shortest paths problem. The problem is also essentially a special case of the problem considered in this paper.

In this paper we present algorithms for the all-pairs version of the minimum-cost plans problem, allowing rechargings along the way. The network may contain negative arc costs. When there are no negative cycles, or when the capacity~$B$ of the battery is sufficiently large, we obtain an $O(n^3)$-time algorithm. This improves and generalizes the $O(n^4\log n)$-time algorithm of Khuller et al.~\cite{khuller2011fill}, which does not allow negative arc costs. Our $O(n^3)$ time bound also applies when at most~$\Delta\le n$ rechargings are allowed, improving the $O(n^3\min\{\Delta\log n,\Delta^2\})$-time bound of Khuller et al.~\cite{khuller2011fill}.

We obtain our new results by showing that finding plans that allow rechargings along the way can be reduced in a simple way to the problem studied by Dorfman et al.~\cite{DKTZ23}, i.e., the problem of finding non-recharging optimal paths. More specifically, the optimal energetic costs found by the algorithms of Dorfman et al.~\cite{DKTZ23} are used to set up a constant number of min-plus products. The results of these min-plus products are used as arc costs in an auxiliary graph on which a standard APSP problem needs to be solved.

% We do this by reducing the problem of finding minimum cost plans to two simpler problems. The first problem is computing optimal energetic paths when recharging is not allowed. (This problem actually comes in two equivalent variants, $\alpha_{B,a}(s,t)$ and $\beta_{B,b}(s,t)$, as explained below.) The second problem is the standard all-pairs shortest paths problem.

The rest of the paper is organized as follows. In the next section we review the results of Dorfman et al.~\cite{DKTZ23} for the no-rechargings case.
% describe the problem of finding optimal energetic paths when no rechargings are allowed and introduce the quantities~$\alpha_{B,a}(s,t)$ and~$\beta_{B,b}(s,t)$. 
In Section~\ref{S:reduction} we show that the problem of finding  plans can be reduced to finding optimal paths when no rechargings are allowed and to computing standard shortest paths. In Section~\ref{S-algs} we use the reduction to obtain our efficient algorithms for finding  plans. In Section~\ref{S:inv-reduction} we describe a simple reduction in the opposite direction, from computing standard shortest paths to computing  plans, showing that the algorithms obtained for finding  plans are essentially optimal. In Section~\ref{S-initial-charges} we consider settings in which the initial charge of the car is not zero. We end in Section~\ref{S-concl} with some concluding remarks and open problems.

% \section{The quantities \texorpdfstring{$\alpha_{B,a}(s,t)$ and $\beta_{B,b}(s,t)$}{alpha(s,t) and beta(s,t)}}\label{S:alpha-beta}
% \section{No rechargings allowed }
\section{No rechargings allowed - the quantities \texorpdfstring{$\alpha_{B,a}(s,t)$ and $\beta_{B,b}(s,t)$}{alpha(s,t) and beta(s,t)}}\label{S:alpha-beta}

We review the no-rechargings results of Dorfman et al.~\cite{DKTZ23} and introduce the quantities $\alpha_{B,a}(s,t)$ and $\beta_{B,b}(s,t)$ that play an important role in the reduction presented in the next section.

Let $\alpha_{B,a}(s,t)$ be the Maximum Final Charge (MFC) with which the car can reach~$t$ if it starts in~$s$ with charge~$a$, where $0\le a\le B$, in its battery, when \emph{no} recharging is allowed along the way. If it is not possible to get from~$s$ to~$t$ with an initial charge of~$a$ at~$s$, we let $\alpha_{B,a}(s,t)=-\infty$. We will mostly be interested in the cases $a=0$ or $a=B$, i.e., starting from~$s$ with either an empty or a full battery.

Let $\beta_{B,b}(s,t)$ be the Minimum Initial Charge (MIC) needed in~$s$ in order to reach~$t$ with a charge of at least~$b$ when again \emph{no} recharging is allowed along the way. If it is not possible to get to~$t$ with a charge of~$b$ even if we start from~$s$ with a charge of~$B$, i.e., with a full battery, we let $\beta_{B,b}(s,t)=\infty$. Again, we are mostly interested in the cases $b=0$ or $b=B$.

If all arc costs are non-negative then it is easy to see that $\alpha_{B,a}(s,t)=a-\delta(s,t)$, if $a\ge \delta(s,t)$, and $\alpha_{B,a}(s,t)=-\infty$ otherwise, where $\delta(s,t)$ is the standard distance from~$s$ to~$t$ in the graph, i.e., the length of the shortest path from~$s$ to~$t$ with respect to the cost function~$c$. Similarly, $\beta_{B,b}(s,t)=\delta(s,t)+b$, if $\delta(s,t)+b\le B$, and $\beta_{B,b}(s,t)=\infty$ otherwise.
Thus, if all arc costs are non-negative, then $\alpha_{B,a}(s,t)$ and $\beta_{B,b}(s,t)$ for all $s,t\in V$ and $a,b\in\{0,B\}$, can be easily computed after solving a standard All-Pairs Shortest Paths (APSP) problem.%\bob{Define APSP.}

When arc costs can be both positive and negative, computing all values of $\alpha_{B,a}(s,t)$ and $\beta_{B,b}(s,t)$ is more difficult. Dorfman et al.~\cite{DKTZ23} gave an $O(mn+n^2\log n)$-time algorithm for the all-pairs versions when there can be negative arcs but there are no negative cycles and an $O(mn^2+n^3\log n)$-time algorithm when there can be negative cycles. They also gave an $O(mn+n^2\log n)$-time algorithm for the all-pairs versions when there can be negative cycles, if the capacity of the battery is sufficiently large, namely $B\ge 3nM$, where $M=\max_{(u,v)\in A}|c(u,v)|$. \footnote{This can be improved to $B\ge nM$.}
%We will hopefully describe a more efficient algorithm in a later section.

% In a previous paper, Dorfman et al. \cite{DKTZ23}

The algorithms of Dorfman et al.~\cite{DKTZ23}, which we do not repeat here, work with another quantity $\delta_{B,a}(s,t)=a-\alpha_{B,a}(s, t)$, %\bob{Shouldn't this be $B-...$?  Also, there is a parameter a in the dfn of $\alpha$, but here it is b.  I suggest using b in both places.} 
% I changed the $b$ to $a$ to make it consistent with the definition of $\alpha_{B,a}(s,t)$. I prefer to keep $b$ for the definition of $\beta_{B,b}(s,t)$. Note that $a$ and $b$ have different meanings in the definitions: $a$ is the initial charge at $s$, $b$ is the required final charge at $t$. These are sufficient different to merit the use of different letters to represent them. The definition used in DKTZ23 is $a-\ldots$ and not $B-\ldots$. 
the \emph{minimum depletion} of the battery when starting at~$s$ with a charge~$a$ in the battery. This quantity is closely related to the standard distance $\delta(s,t)$ from~$s$ to~$t$ in the graph and in some cases equal to it. They show that $\delta_{B,a}(s,t)$ can be computed by suitable adaptations of the standard Bellman-Ford and Dijkstra algorithms. They also show that $\beta^G_{B,b}(s,t)=B-\alpha_{B,B-b}^{\bvec{G}}(t,s)$, where $\bvec{G}$ is the graph obtained from~$G$ by reversing the direction of all arcs, maintaining the arc costs. Thus, algorithms for computing $\alpha_{B,a}(s,t)$ can also be used to compute $\beta_{B,b}(s,t)$, and vice versa. (This is true for the all-pairs version of the problems. A single-source algorithm for the maximum final charges $\alpha_{B,a}(s,t)$ becomes a single-target algorithm for the minimum initial charges $\beta_{B,b}(s,t)$.)

% \section{Reducing \texorpdfstring{$\rho_B(s,t)$ to $\alpha_{B,a}(s,t)$, $\beta_{B,b}(s,t)$ and $\delta(s,t)$}{rho(s,t) to alpha(s,t) and beta(s,t)}}\label{S:reduction}

% \section{The reduction}\label{S:reduction}
\section{Handling rechargings by reduction to no rechargings}\label{S:reduction}

In this section we describe a simple reduction from computing $\rho_B(s,t)$, for every $s,t\in V$, to the computation of $\alpha_{B,a}(s,t)$ and $\beta_{B,b}(s,t)$, for every $s,t\in V$ and $a,b\in\{0,B\}$, and to the solution of a standard all-pairs shortest paths (APSP) problem on an auxiliary network.

No assumptions are made in this section on $c:A\to \RR$. Arc costs may be negative and the graph may even contain negative cycles. ($\rho_B(s,t)$ is well-defined even in the presence of negative cycles.) For simplicity we assume that all finite charging costs $r(v)$ are distinct. (Ties can be broken easily.)

\subsection{Travel plans}

We begin with a formal definition of \emph{travel plans}. We use $u^a$, where $u\in V$ and $a\in[0,B]$ to denote the state in which the car is at~$u$ with a charge of~$a$ in its battery. A travel plan is then simply a sequence of states such that the move between two consecutive states corresponds to either traversing an arc of the graph, or recharging the battery by a certain amount.

\begin{Definition}[Travel plans] A \emph{travel plan} $P$ from~$s$ to~$t$ in a graph $G=(V,A,c,r)$ is a sequence $u_0^{a_0} u_1^{a_1}\ldots u_\ell^{a_\ell}$, where $s=u_0,u_1,\ldots,u_\ell=t\in V$, $0=a_0$, $0 \leq a_i \leq B$ for all $0 \leq i \leq \ell$, and such that for every $i=0,1,\ldots,\ell-1$, either $(u_i,u_{i+1})\in A$, $a_i\ge c(u_i,u_{i+1})$ and $a_{i+1}=\min\{a_i-c(u_i,u_{i+1}),B\}$, corresponding to legally traversing an arc, or $u_i=u_{i+1}$ and $a_i<a_{i+1}$, corresponding to a charging of the battery. The cost of the plan is $cost(P)=\sum_{i\in R} r(u_i)(a_{i+1}-a_i)$, where $R=\{0\le i\le\ell\mid u_i=u_{i+1}\}$ is the set of indices in which rechargings take place.
\end{Definition}

If $P=u_0^{a_0} u_1^{a_1}\ldots u_\ell^{a_\ell}$ is a travel plan, then $u_0,u_1,\ldots,u_\ell$ is the path used by the plan, where each vertex in which a recharging takes place has two consecutive appearances. A minimum-cost plan from~$s$ to~$t$ is of course a plan from~$s$ to~$t$ whose cost is minimum. If the graph contains no negative cycles, it is not difficult to see that a path used by an optimal plan may be assumed to be simple. This is not true, however, in the presence of negative cycles.

\newcommand{\G}{{\cal G}}
\newcommand{\V}{{\cal V}}
\newcommand{\E}{{\cal E}}

Conceptually, we can define an infinite graph $\G=(\V,\E,\ell)$, where $\V=\{u^a \mid u\in V\;,\; a\in[0,B]\}$, $\E=\E_1\cup \E_2$, $\E_1=\{(u^a,v^b)\mid (u,v)\in A\;,\;a\ge c(u,v)\;,\; b=\min\{a-c(u,v),B\}$, $\E_2=\{(u^a,u^b)\mid u\in V\;,\; a<b\}$, and where $\ell(u^a,v^b)=0$, for every $(u^a,v^b)\in \E_1$, and $\ell(u^a,u^b)=r(u)(b-a)$, for every $(u^a,u^b)\in \E_2$. Then a travel plan is simply a path in~$\G$, and the cost of the plan is simply the length of the path with respect to the length function~$\ell$. The challenge, of course, is to efficiently compute shortest paths in this implicit infinite graph. (If all arc costs are integral, the graph becomes finite, but its size is $O(B(m+n))$ which is still too large.)

% We occasionally partition a plan~$P$ into segments $P=P_1P_2\ldots P_k$, where $P_i=x_i^{\bar{a}_i}\ldots y_i^{\bar{b}_i}$. More formally, if $0=\ind_1<\ind_2<\cdots<\ind_k$ are the indices of the first items in the segments $P_1,P_2,\ldots,P_k$ and $\ind_{k+1}=\ell+1$, then $x_i=u_{\ind_i}$, $\bar{a}_i=a_{\ind_i}$, $y_i=u_{\ind_{i+1}-1}$ and $\bar{b}_i=a_{\ind_{i+1}-1}$.

\subsection{Structure of optimal plans}

% If $P$ is a path and $x,y$ are vertices on~$P$, with~$x$ appearing before~$y$, we let $P[x,y]$ denote the portion of the path~$P$ from~$x$ to~$y$. We also define $P[x,y)$ and $P(x,y]$ to be the portions in which~$y$ or~$x$ are excluded, respectively.

% We say that a path~$P$ from~$s$ to~$t$ is a minimum cost path \haim{Italize ? this is a definition} if there is a minimum cost travel plan from~$s$ to~$t$ that uses~$P$. \haim{We did not formally define the path used by a travel plan} Once a minimum cost path from~$s$ to~$t$ is identified it is easy to find an optimal recharging schedule.  \haim{Is this obvious ?}

% \begin{lemma}\label{L-x-old}[Old]
%     Let $P$ be a minimum cost path from $s$ to $t$. Let $x_1,x_2,\ldots,x_{k}$ be the vertices on~$P$ in which recharging takes place. (If the battery is charged at~$s$ then $x_1=s$.) Then, for $i=1,\ldots,k-1$,
%     \begin{itemize}
%     \item[\textup{(i)}] If $r(x_i)> r(x_{i+1})$, then there is a vertex~$y_{i}\in P(x_i,x_{i+1}]$ that the car reaches with an empty battery. (The vertex~$y_{i}$ may be~$x_{i+1}$.)
%     \item[\textup{(ii)}] If $r(x_i)<r(x_{i+1})$, then there is a vertex~$y_{i}\in P[x_i,x_{i+1})$ that the car leaves with a full battery. (The vertex~$y_{i}$ may be~$x_i$.)
%     \end{itemize}
%\end{lemma}

Each plan $P$ from~$s$ to~$t$ can be partitioned into segments $P=P_0P_1\ldots P_k$ such that no rechargings take place within each segment and such that a recharging does take place when moving from segment to segment. Each segment $P_i$ is then of the form $P_i=x_i^{b_i}\cdots x_{i+1}^{c_i}$, where $c_i<b_{i+1}$, if $i<k$. Also $s=x_0$ and $t=x_{k+1}$. The whole plan~$P$ is then $x_0^{b_0}\cdots x_1^{c_0}|x_1^{b_1}\cdots x_2^{c_1}|x_2^{b_2}\cdots x_3^{c_2}|\cdots|x_{k-1}^{b_{k-1}}\cdots x_k^{c_{k-1}}|x_k^{b_k}\cdots x_{k+1}^{c_k}$, where each vertical line separates two segments and indicates a place at which a charging takes place. The vertices $x_1,x_2,\ldots,x_k$ are the vertices in which rechargings take place. (If no rechargings take place along~$P$, then $k=0$ and $P=P_0$. If a charging takes place at~$s$, then $s=x_0=x_1$ and $P_0=s^0$.)

The following two lemmas are extensions of similar lemmas of Khuller, Malekian and Mestre~\cite{khuller2011fill}.

\begin{lemma}\label{L-x}
    Let $P$ be a minimum-cost plan from $s$ to $t$. Let $P=P_0P_1\ldots P_k$ a partition of~$P$ into segments as above, and let $P_i=x_i^{b_i}\cdots x_{i+1}^{c_i}$, where $c_i<b_{i+1}$, if $i<k$. Recall that $x_1,x_2,\ldots,x_k$ are the vertices on~$P$ in which rechargings take place and that $s=x_0$ and $t=x_{k+1}$. Then, for $i=1,\ldots,k-1$,
    \begin{itemize}
    \item[\textup{(i)}] If $r(x_i)> r(x_{i+1})$, then $P_i$ contains a state~$y_{i}^0$, i.e., a vertex~$y_i$ at which the car has an empty battery. (The vertex~$y_{i}$ may be~$x_{i+1}$ but not~$x_i$.)
    \item[\textup{(ii)}] If $r(x_i)<r(x_{i+1})$, then $P_i$ contains a state~$y_{i}^B$, i.e., a vertex~$y_i$ at which the car has a full battery. (The vertex~$y_{i}$ may be~$x_i$ but not~$x_{i+1}$.)
    \end{itemize}
\end{lemma}

\begin{proof}
Let $1\le i<k$ and let $P_i=u_1^{a_1}u_2^{a_2}\ldots u_\ell^{a_\ell}$, where $u_1^{a_1}=x_i^{b_i}$ and $u_\ell^{a_\ell}=x_{i+1}^{c_i}$. Let $d_i=b_i-c_{i-1}>0$ and $d_{i+1}=b_{i+1}-c_{i}>0$ be the amount of charge added to the battery at~$x_i$ and $x_{i+1}$, respectively.

    (i) Assume that $r(x_i)> r(x_{i+1})$. Let $\varepsilon=\min\{a_1,a_2,\ldots,a_\ell\}$. If $\varepsilon=0$, we are done, as there is an index $1\le j\le \ell$ such that $a_j=0$. (Note that $j\ge 2$ as $a_1>0$.) Otherwise, let $\delta=\min\{d_i,\varepsilon\}>0$. We can then replace the segment~$P_i$ in~$P$ with the segment $P'_i=u_1^{a_1-\delta}u_2^{a_2-\delta}\ldots u_\ell^{a_\ell-\delta}$. This is equivalent to charging~$\delta$ less at~$x_i$ and~$\delta$ more at~$x_{i+1}$. The resulting plan~$P'$ is valid and $cost(P')=cost(P)-\delta\,(r(x_i)-r(x_{i+1}))<cost(P)$, contradicting the optimality of~$P$.

    (ii) Assume that $r(x_i)< r(x_{i+1})$. Let $\varepsilon=B-\max\{a_1,a_2,\ldots,a_\ell\}$. If $\varepsilon=0$, we are done, as there is an index $1\le j\le \ell$ such that $a_j=B$. (Note that $j<\ell$ as $a_\ell<B$.) Otherwise, let $\delta=\min\{d_{i+1},\varepsilon\}>0$. We can then replace the segment~$P_i$ in~$P$ with the segment $P'_i=u_1^{a_1+\delta}u_2^{a_2+\delta}\ldots u_\ell^{a_\ell+\delta}$. This is equivalent to charging~$\delta$ more at~$x_i$ and~$\delta$ less at~$x_{i+1}$. The resulting plan~$P'$ is valid and $cost(P')=cost(P)+\delta\,(r(x_i)-r(x_{i+1}))<cost(P)$, again contradicting the optimality of~$P$.
\end{proof}

See Figure~\ref{F-fig1} for a schematic description of Lemma~\ref{L-x} and the following Lemma~\ref{L-y} which follows almost immediately from Lemma~\ref{L-x}.

% \begin{lemma}[Old]\label{L-y-old}
%     Let $P$ be a minimum cost path from $s$ to $t$. Then, there are vertices $s=y_0,\ldots,y_{k}=t$ on~$P$ such that $s=y_0,\ldots,y_{k-1}$ are reached with either an empty or a full battery and such that at most one recharging takes place between~$y_i$ and~$y_{i+1}$, for $i=0,\ldots,k-1$. 
% \end{lemma}

\begin{lemma}\label{L-y}
    Let $P$ be a minimum-cost plan from~$s$ to~$t$. Then, there are states $y_0^{a_0},y_1^{a_1}\ldots,y_{k}^{a_k}$ on~$P$ such that $s=y_0$, $t=y_k$, $a_0=0$, and $a_i\in\{0,B\}$, for $i=1,\ldots,k-1$, and such that at most one recharging takes place between~$y_i^{a_i}$ and~$y_{i+1}^{a_{i+1}}$ on~$P$, for $i=0,\ldots,k-1$. (In other words $y_i$ is either reached with an empty battery, or left with a full battery, for each $i=0,1,\ldots,k-1$. Note that $a_k$, the final charge at~$y_k=t$, is not required to be in $\{0,B\}$.)
\end{lemma}

\begin{proof}
    Let $P=P_0P_1\ldots P_{k}$ be the partition of~$P$ as in Lemma~\ref{L-x}. By the lemma, each segment $P_i$, for $i=1,2,\ldots,k-1$, contains a state $y_i^{a_i}$, where $a_i\in\{0,B\}$. Let $y_0^{a_0}=s^0$ and $y_k^{a_k}=t^{a_k}$, where~$a_k$ is the final charge at~$t$ according to~$P$. The sequence $y_0^{a_0},y_1^{a_1}\ldots,y_{k}^{a_k}$ satisfies all the required properties. In fact, there is exactly one recharging between~$y_i^{a_i}$ and~$y_{i+1}^{a_{i+1}}$ on~$P$, for $i=0,\ldots,k-1$, unless $k=0$, in which case $y_0^{a_0}=s^0$, $y_1=t^{a_1}$ and there is no recharging between them.
\end{proof}

% \begin{proof} (Old)
%     Let $x_1,x_2,\ldots,x_{k}$ be the vertices on~$P$ in which the battery is recharged. (If no recharging is used, then $k=0$ and some of the claims below hold vacuously.) By Lemma~\ref{L-x}, there is a vertex~$y_{i}$ between $x_i$ and $x_{i+1}$, where $i=1,\ldots,k-1$, that the car reaches with either an empty or a full battery. At most one recharging, at $x_{i+1}$, takes place between~$y_{i}$ and~$y_{i+1}$, for $i=1,\ldots,k-2$. We let $y_0=s$ and $y_k=t$. Since $x_{k-1}$ is the last recharging stop on the path, at most one recharging takes place between $y_{k-1}$ and $y_k=t$. Also, only one recharging, at $x_1$ takes place between $y_0=s$ and~$y_1$. (Note that this holds whether or not $x_1=s$, see Figure~\ref{F-fig3}.)
% \end{proof}

We note that in the sequence $y_0^{a_0},y_1^{a_1}\ldots,y_{k}^{a_k}$ whose existence is proved in Lemma~\ref{L-y} we may have $y_i=y_{i+1}$, in which case $a_i=0$ and $a_{i+1}=B$. This can happen if $y_i^{a_i}$ is the last state in~$P_i$ and $y_{i+1}^{a_{i+1}}$ is the first state in~$P_{i+1}$.

% It should be pointed out that adjacent vertices in the sequence $s=y_0,\ldots,y_{k}=t$ whose existence was proved in Lemma~\ref{L-y} may be equal. If, for example, $r(x_{i-1})>r(x_i)<r(x_{i+1})$, we may have % $y_i=y_{i+1}=x_i$ \dani{$y_{i-1}=y_i=x_i$?}, $y_{i-1}=y_i=x_i$, i.e., the car reaches~$x_i$ with an empty battery and leaves $x_i$ with a full battery. Our use of Lemma~\ref{L-y} takes into account this possibility. Non-adjacent vertices in the sequence may be assumed to be distinct.

If there are no negative arc costs, then a minimum-cost plan from~$s$ to~$t$ always reaches~$t$ with an empty battery. This is not necessarily the case when arc costs may be negative.

\begin{figure}[t]
\begin{center}
\includegraphics[scale=0.43]{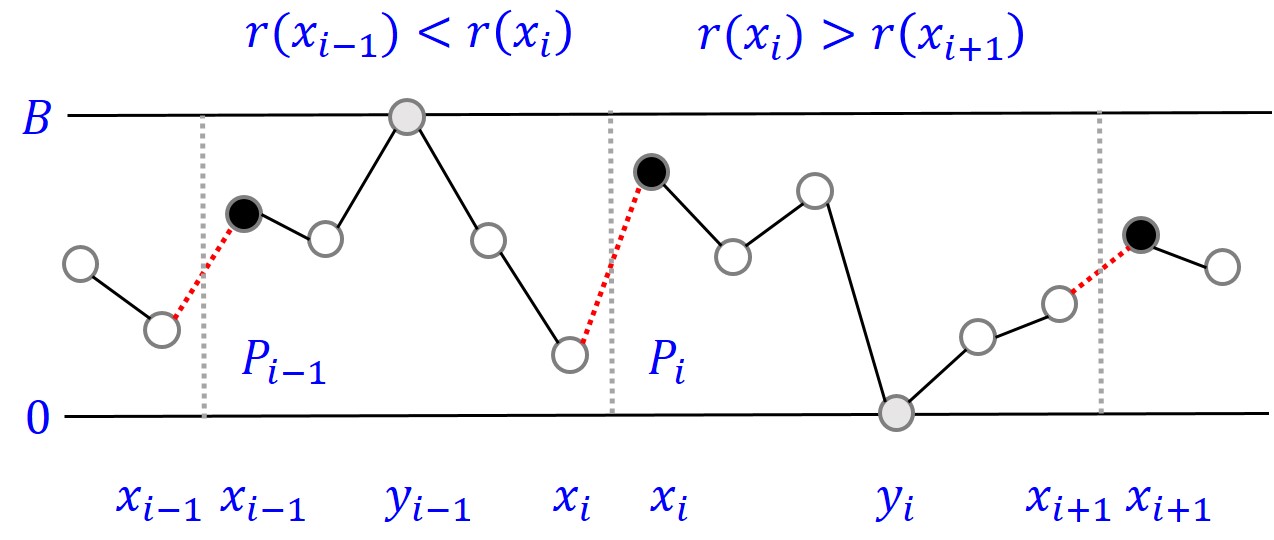} 
\end{center}
\vspace*{-15pt}
\caption{A schematic description of Lemmas~\ref{L-x} and~\ref{L-y}. The height of a vertex represents the charge in the battery. Dotted red lines represent rechargings.}
\label{F-fig1}
%\vspace*{-5pt}
\end{figure}

% \begin{figure}[t]
% \begin{center}
% \includegraphics[scale=0.35]{Figures/figure3.jpg} 
% \end{center}
% \vspace*{-5pt}
% \caption{A schematic description of the proof of Lemma~\ref{L-y}. Filled vertices are vertices in which charging takes place. At the top, charging takes place at the start vertex~$s$. At the bottom charging does not take place at~$s$. In both cases, the car starts at~$s$ with an empty battery.}
% \label{F-fig3}
% \vspace*{-5pt}
% \end{figure}

% See Figure~\ref{F-fig3} for a schematic description of the proof of Lemma~\ref{L-y}.

% The algorithm of Section~\ref{sub:non-rho} for the non-negative case focuses on optimal sub-paths from $x_i$ and $x_{i+1}$ where charging is only allowed at~$x_i$. Here it seems more helpful to consider optimal paths from $y_i/z_i$ to $y_{i+1}/z_{i+1}$ where only one recharging is allowed midway (at $x_{i+1}$) and where we leave $y_i$ with an empty battery, leave $z_i$ with a full battery, and need to reach $z_{i+1}$ with a full battery.

\subsection{The reduction}\label{S-reduction}

Lemma~\ref{L-y} suggests the following algorithm for computing $\rho_B(s,t)$ for every $s,t\in V$. Construct a complete graph $G^{0,B}$ on the vertex set $V^{0,B}=V^0\cup V^B$, where $V^a=\{ u^a \mid u\in V\}$, for $a\in\{0,B\}$. Here, $u^a$ corresponds to being at~$u$ with charge of at least~$a$ in the battery. We define a new cost function~$\ell_B$ such that $\ell_B(u^a,v^b)$, where $u,v\in V$ and $a,b\in\{0,B\}$, is the minimum cost of getting from~$u$ to~$v$, starting with an initial charge of~$a$ at~$u$ and reaching~$v$ with a charge of at least~$b$, when \emph{at most one recharging} is allowed along the way. By Lemma~\ref{L-y}, computing minimum-cost plans in the original graph would then reduce to finding shortest paths in~$G^{0,B}$ with respect to the cost function $\ell_B$. More specifically, $\rho_B(s,t)=\delta_{\ell_B}(s^0,t^0)$, for every $s,t\in V$, where $\delta_{\ell_B}(s^0,t^0)$ is the distance from~$s^0$ to~$t^0$ with respect to~$\ell_B$. (Note that $t^0$ here stands for reaching~$t$ with a charge of at least~$0$. When there are negative arc costs, a minimum-cost plan may reach the destination with strictly positive charge.)

% For each vertex $x\in V$ we define two new vertices $x^0$ and $x^B$ that represent the cases in which we start or reach~$x$ with either an empty or a full battery.

We next describe how the new costs $\ell_B(u^a,v^b)$ are computed. This is the main step of our reduction. Consider the cheapest way of getting from $u^a$ to~$v^b$, i.e., from~$u$ with initial charge~$a$ to~$v$ with final charge at least~$b$, with (an optional) recharging at a given vertex $x\in V$ along the way. Recharging is not allowed in any other vertex. (We do not assume that $u,v$ and~$x$ are distinct. We may even have $u=v=x$.) Clearly, we should choose a path from~$u^a$ to~$x$ that reaches~$x$ with the maximum possible charge, namely $\alpha_{B,a}(u,x)$. We then buy just enough charge at~$x$ to reach~$v$ with a charge of at least~$b$. The minimum initial charge at~$x$ required to reach~$v^b$ is $\beta_{B,b}(x,v)$. The amount of charge we need to buy at~$x$ is thus $(\beta_{B,b}(x,v)-\alpha_{B,a}(u,x))^+$, where $z^+=\max\{0,z\}$. (If $\alpha_{B,a}(u,x)>\beta_{B,b}(x,v)$, we do not need to recharge at~$x$.) Considering all choices for the recharging vertex~$x$, including~$u$ and~$v$, we get:
\[ \ell(u^a,v^b) \EQ \min_{x\in V} \; r(x)\cdot(\beta_{B,b}(x,v)-\alpha_{B,a}(u,x))^+ \;. \]
We allow the case $x=u$, in which case $\alpha_{B,a}(u,u)\ge a$. (Strict inequality is possible in the presence of negative cycles.) We also allow the case $x=v$, in which case $\beta_{B,b}(v,v)\le b$. (Again a strict inequality is possible in the presence of negative cycles.)

To compute $\ell(u^a,v^b)$ we first remove the $+$ from the definition of $\ell(u^a,v^b)$, namely,
\[ \ell'(u^a,v^b) \EQ \min_{x\in V} \; r(x)\cdot(\beta_{B,b}(x,v)-\alpha_{B,a}(u,x)) \;. \]
It is then easy to see that
\[ \ell(u^a,v^b) \EQ \ell'(u^a,v^b)^+\;,\]
as if $\ell'(u^a,v^b)<0$ then it is possible to get from $u^a$ to $v^b$ without recharging, i.e., at cost~$0$.

To compute $\ell'(u^a,v^b)$ we form a layered graph $(V^{0,B}\times U)\cup(U \times V^{0,B})$, where $U=\{x\in V \mid r(x)<\infty\}$. (We assume that different copies of $V^{0,B}$ are used in the first and third layers of this graph, see Figure~\ref{F-fig2}.) We let 
\[w(u^a,x)\EQ -r(x)\alpha_{B,a}(u,x) \qquad,\qquad w(x,v^b) \EQ r(x)\beta_{B,b}(x,v)\;.\] 
Now $\ell'(u^a,v^b)$ can be computed using a single min-plus product.\footnote{The min-plus product of an $n\times p$ matrix $A=(a_{i,j})$ and a $p \times n$ matrix $A=(a_{i,j})$, is the $n\times n$ matrix $C=(c_{i,j})$ defined by $c_{i,j}=\min_{k=1}^p a_{i,k}+b_{k,j}$, for every $1\le i,j\le n$.} The time required is $O(n^3)$ using the na\"{i}ve algorithm, or slightly faster using the algorithm of Williams \cite{Williams18,Williams21}, assuming that all $\alpha_{B,a}(u,x)$ and $\beta_{B,b}(x,v)$ values are given to us.

\begin{figure}[t]
\begin{center}
\includegraphics[scale=0.35]{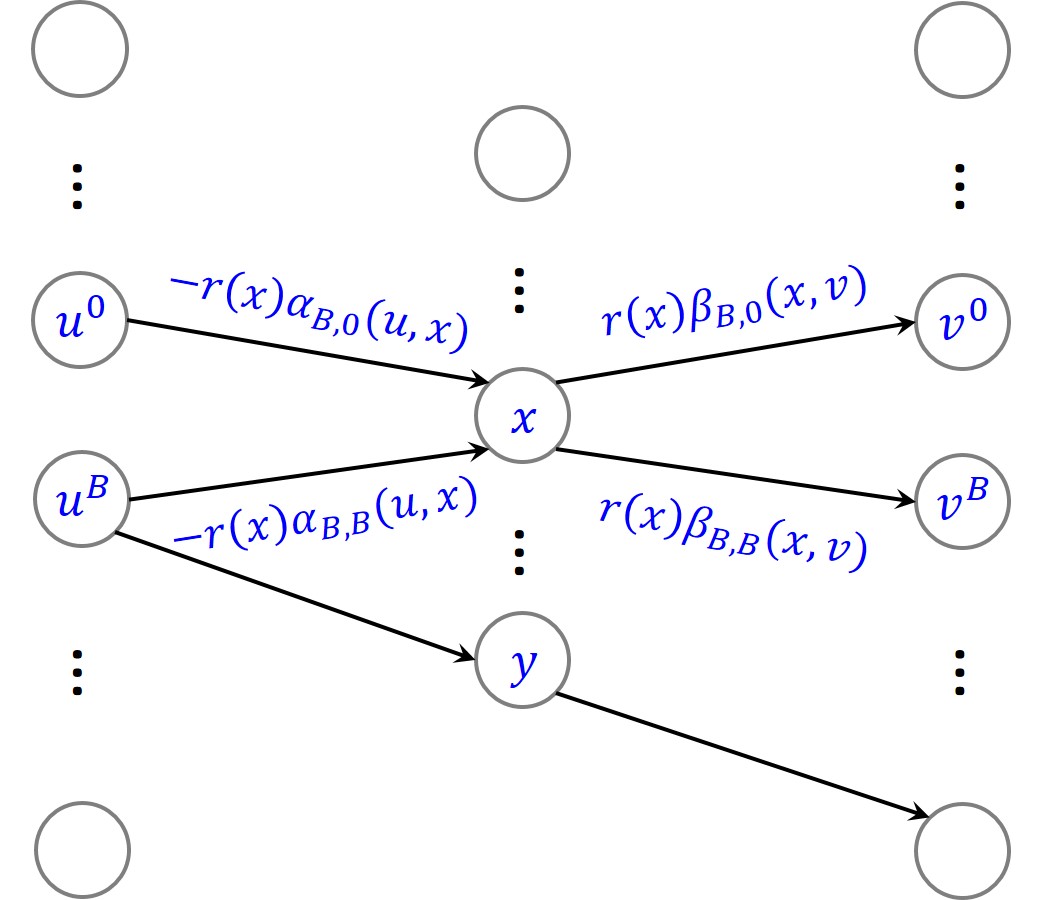} 
\end{center}
%\vspace*{-10pt}
\caption{The min-plus product used to compute $\ell'(u^a,v^b)$, for every $u,v\in V$ and $a,b\in\{0,B\}$. The vertices in the middle layer are the vertices~$x\in V$ for which $r(x)<\infty$.}
\label{F-fig2}
%\vspace*{-10pt}
\end{figure}

The following lemma formally proves the correctness of the reduction. Recall that $\delta_{\ell_B}(s^0,t^0)$ is the distance from~$s^0$ to~$t^0$ in the graph $G^{0,B}$ with respect to~$\ell_B$.

\begin{lemma}
    $\rho_B(s,t)=\delta_{\ell_B}(s^0,t^0)$, for every $s,t\in V$.
\end{lemma}

\begin{proof}
    Any path~$P$ from~$s^0$ to~$t^0$ in the complete graph $G^{0,B}$ with the cost function $\ell_B$ corresponds to a plan $P'$ from~$s$ to~$t$ in the original graph~$G$ whose cost is exactly~$\ell_B(P)$. More specifically, let $P$ be the path $s=y_0^{0},y_1^{a_1},\ldots,y_{k-1}^{a_{k-1}},y_k^0=t$. Then each arc $(y_i^{a_i},y_{i+1}^{a_{i+1}})$ can be replaced by a plan in the original graph~$G$ whose cost is~$\ell(y_i^{a_i},y_{i+1}^{a_{i+1}})$. Concatenating all these plans we get a plan~$P'$ in~$G$ whose total cost is $\ell_B(P)=\sum_{i=0}^{k-1}\ell(y_i^{a_i},y_{i+1}^{a_{i+1}})$. Thus, $\rho_B(s,t)\le \delta_{\ell_B}(s^0,t^0)$.

    Conversely, let $P'$ be a minimum-cost plan from~$s$ to~$t$ in the original graph~$G$. By Lemma~\ref{L-y} there are states $y_0^{a_0},y_1^{a_1},\ldots,y_k^{a_k}$ on~$P$ such that $y_0^{a_0}=s^0$, $y_k=t$, $a_0,\ldots,a_{k-1}\in \{0,B\}$, and such that at most one recharging takes place between~$y_i^{a_i}$ and~$y_{i+1}^{a_{i+1}}$. Let~$P$ be the path $s=y_0^{0},y_1^{a_1},y_2^{a_2}\ldots,y_{k}^{0}$ in $G^{0,B}$. (Note that we replaced $a_k$ here by~$0$, possibly throwing away extra charge with which the plan reaches~$t$.) The cost $\ell_B(P)$ of~$P$ is at most the cost of the plan~$P'$. Thus, $\rho_B(s,t)\ge \delta_{\ell_B}(s^0,t^0)$.
    % Combining these two inequalities we get that $\rho_B(s,t)=\delta_{\ell_B}(s^0,t^0)$, for every $s,t\in V$.
\end{proof}

Let $\MCP(m,n)$ denote the the problem of computing all-pairs \emph{Minimum-Cost Plans} on a graph with~$n$ vertices and~$m$ arcs, i.e., the computation of $\rho_B(s,t)$ for every $s,t\in V$. Let $\MFC(m,n)$ denote the problem of computing all-pairs \emph{Maximum Final Charges}, i.e., the computation of $\alpha_{B,a}(s,t)$ for every $s,t\in V$, for a given $0\le a\le B$, on a graph with~$n$ vertices and~$m$ arcs. Let $\MIC(m,n)$ denote the problem of computing all-pairs \emph{Minimum Initial Charges}, i.e., the computation of $\beta_{B,b}(s,t)$ for every $s,t\in V$, for a given $0\le b\le B$, again on a graph with~$n$ vertices and~$m$ arcs. Let $\MPP(n,p)$ be the problem of computing a \emph{Min-Plus Product} of an $n\times p$ matrix and a $p\times n$ matrix. Finally, let $\APSP(m,n)$ be the problem of computing the All-Pairs Shortest Paths, i.e., $\delta(s,t)$ for every $s,t\in V$, on a graph with $n$ vertices and $m$ edges.

Putting everything together, we get:

\begin{theorem}
    $\MCP(m,n)$ can be reduced to two instances of $\MFC(m,n)$ and $\MIC(m,n)$ each, an instance of $\MPP(2n,p)$, where $p\le n$ is the number of charging stations, and an instance of $\APSP(4n^2,2n)$.
\end{theorem}

We need two instances of $\MFC(m,n)$, for $a=0,B$, and two instances of $\MIC(m,n)$, for $b=0,B$. As noted, $\MIC(m,n)$ can be reduced to $\MFC(m,n)$ by simply reversing the graph. It is also interesting to note that $APSP(n^2,n)$ can be reduced to a collection of $\MPP(n,n)$ problems that can be solved in the time needed for a single $\MPP(n,n)$ instance (see Aho et al.~\cite{AHU74}). Also, $\MPP(n,p)$ can easily be reduced to $APSP(2np,2n+p)$.

Similarly, if we let $\MCP_\Delta(m,n)$ be the problem of computing $\rho_{B,\Delta}(s,t)$ for every $s,t\in V$, i.e., computing minimum-cost plans between every pair of vertices when at most~$\Delta$ rechargings are allowed, and $\APSP_\Delta(u,v)$ be the problem of computing all-pairs shortest paths that are allowed to contain at most~$\Delta$ arcs, then we have 

\begin{theorem}
    $\MCP_\Delta(m,n)$ can be reduced to two instances of $\MFC(m,n)$ and $\MIC(m,n)$ each, an instance of $\MPP(2n,p)$, where $p\le n$ is the number of charging stations, and an instance of $\APSP_\Delta(4n^2,2n)$.
\end{theorem}

% \section{Algorithms for \texorpdfstring{$\rho_B(s,t)$}{rho(s,t)}}\label{S-algs}
\section{Algorithms for minimum-cost plans}\label{S-algs}

Using the reduction of the previous section and the best available algorithms for the Maximum Final Charge (\MFC\,), Minimum Initial Charge (\MIC\,), Min-plus product (\MPP\,) and All-Pairs Shortest Paths (\APSP\,) we obtain the following results:

\begin{theorem}\label{T-no}
    The all-pairs version of the Minimum-Cost Plans (\MCP) problem in a graph with no negative cycles can be solved in $O(\frac{n^3}{2^{c\sqrt{\log n}}}+mn)=O(n^3)$ time, for some $c>0$. The same time bound, with a different constant~$c$, applies when at most $\Delta\le n$ rechargings can be used on each path.
\end{theorem}

\begin{proof}
    Dorfman et al.~\cite{DKTZ23} gave algorithms for solving the \MFC\ and the \MIC\ problems in $O(mn+n^2\log n)$ time, on graphs with no negative cycles. The \MPP\ and \APSP\ problems can be solved in $O(\frac{n^3}{2^{c\sqrt{\log n}}})$ time using a randomized algorithm of Williams \cite{Williams18,Williams21}, or a deterministic algorithm of Chan and Williams \cite{ChWi21}, for some $c>0$.

    When a bound $\Delta$ is placed on the number of chargings, we need to solve an $\APSP_\Delta(n^2,n)$ problem, rather than an $\APSP(n^2,n)$ problem. It is well-known that such a problem can be solved by computing $O(\log\Delta)$ \MPP\ problems, \footnote{We simply need to raise the weighted adjacency matrix of the graph, with zeros on the diagonal, to the $\Delta$-th power with respect to min-plus products. This can be easily done using $O(\log\Delta)$ min-plus products.} and hence can also be solved in $O(\frac{n^3}{2^{c\sqrt{\log n}}})$ time, for some $c>0$. (Note that $\frac{n^3\log n}{2^{c\sqrt{\log n}}}=O(\frac{n^3}{2^{c'\sqrt{\log n}}})$ for any $c'<c$.)
\end{proof}

The bound in Theorem~\ref{T-no} also applies when there may be negative cycles, if the capacity~$B$ of the battery is sufficiently large, i.e., $B\ge nM$, where $M=\max_{(u,v)\in A}|c(u,v)|$. (For simplicity, we did not include this in the statement of the theorem.)

We note that the instance of the \APSP\/ problem that we need to solve is dense, i.e., $m=n^2$. Thus, efficient algorithms for sparse instances of the \APSP\ problem, such as the $O(mn+n^2\log n)$-time algorithm obtained by running Dikstra's \cite{Dijkstra59} algorithm from every vertex, implemented using Fibonacci heaps \cite{FrTa87} or Hollow heaps \cite{HKTZ17}, or the slightly faster $O(mn+n^2\log\log n)$-time algorithm of Pettie~\cite{Pettie04}, would both require $O(n^3)$ time. We also note that a slightly faster algorithm for solving the $\APSP_\Delta(n^2,n)$ problem can be obtained using the sampling technique of Zwick \cite{Zwick02}, showing that when $\Delta>\log n$, the $O(\log \Delta)$ min-plus products can essentially be replaced by only $O(\log\log n)$ such products. This does not change the value of the constant~$c$. 

\begin{theorem}\label{T-yes}
    The all-pairs version of the Minimum-Cost Plan (\MCP) problem in a graph that may contain negative cycles can be solved in $O(mn^2+n^3\log n)$ time. The same time bound applies when at most $\Delta$ rechargings can be used on each path.
\end{theorem}

\begin{proof}
    Dorfman et al.~\cite{DKTZ23} gave algorithms for solving the all-pairs \MFC\ and \MIC\ problems in $O(mn^2+n^3\log n)$ time on graphs that may contain negative cycles. The \MPP\ problem can be solved in $O(n^3)$ time using the na\"{i}ve algorithm and the $\APSP(n^2,n)$ problem can be solved in $O(n^3)$ time using the Floyd-Warshall algorithm \cite{Floyd62,Warshall62}. The $\APSP_\Delta(n^2,n)$ can be solved in $O(n^3\log\Delta)=O(n^3\log n)$ time by solving $O(\log \Delta)$ \MPP\ problems.
\end{proof}

All the time bounds above are in a model similar to the addition-comparison model used for standard shortest paths algorithms. (See, e.g., Zwick \cite{Zwick01}.) Arc and recharge costs can be arbitrary real numbers but it is assumed that `reasonable' operations on them can be performed in constant time. Faster algorithms can be obtained if the battery capacity~$B$ is a relatively small integer, all arc costs $c(a)$ are integers, and all charging costs $r(v)$ are fairly small integers. In particular, if $R=\max\{r(v)\mid v\in V \,,\,r(v)<\infty\}$, then all arc lengths in the $\APSP(n^2,n)$ instance that needs to be solved are in $\{0,1,\ldots,BR\}$. Zwick \cite{Zwick02} obtained an $O(M^{1/(4-\omega)}n^{2+1/(4-\omega)})$-time algorithm for the APSP problem when arc lengths are integers of absolute value at most~$M$, where $\omega$ is the exponent of matrix multiplication. Currently $\omega<2.3719$ due to a recent result of Duan, Wu and Zhou \cite{DuWuZh22}. Note that the running time of this algorithm is subcubic when $M\ll n^{3-\omega}$, or in our case when $BR\ll n^{3-\omega}$. If $\omega>2$, then an improved running time can be obtained using rectangular matrix multiplication. See Zwick \cite{Zwick02} for the exact details.

Faster algorithms can also be obtained for the \MFC\ and \MIC\ problems in the word-RAM model, in which arc costs are assumed to be integers that fit into single machine words. This is done by replacing the standard comparison-based priority queues used by these algorithms by word-RAM priority queues. (See, e.g., Thorup~\cite{Thorup04}.) The running times of the improved algorithms, when there are no negative cycles, are still $\Omega(mn)$. Thus, no improvement is obtained over the bound given in Theorem~\ref{T-no}.

% \section{Reducing \texorpdfstring{$\delta(s,t)$ to $\rho_B(s,t)$}{delta(s,t) to rho(s,t)}}\label{S:inv-reduction}
\section{Reducing shortest paths to minimum-cost plans}\label{S:inv-reduction}

In this short section we describe a trivial reduction from $\APSP(m,n)$ to $\MCP(m+n,2n)$, showing that the Minimum-Cost Plans (\MCP\,) problem is at least as hard as the classical All-Pairs Shortest Paths (\APSP\,) problem.

\begin{theorem}
    The $\APSP(m,n)$ problem in a graph with nonnegative arc costs can be reduced in linear time to an $\MCP(m+n,2n)$ problem.
\end{theorem}

\begin{proof}
    Let $G=(V,A,c)$, where $c:A\to\RR^+$ be the input to the \APSP\ problem. Construct a graph $G'=(V',A',c',r')$ where $V'=\{v',v\mid v\in V\}$ and $A'=\{(v',v)\mid v\in V\}\cup A$. Let $c'(v',v)=0$, for every $v\in V$, and $c'(a)=c(a)$, if $a\in A$. Also, let $r(v')=1$ and $r(v)=\infty$, for every $v\in V$. Let $B=n\max_{a\in A}c(a)$. We have $\delta^G(s,t)=\rho_B^{G'}(s',t)$, for every $s,t\in V$, where $\delta^G(s,t)$ is the distance from~$s$ to~$t$ in~$G$ and $\rho_B^{G'}(s',t)$ is the minimum cost of a plan from~$s'$ to~$t$ in~$G'$. This follows since in~$G'$ it is only possible to charge the battery at the initial vertex and all arc costs are non-negative. Thus, a  plan corresponds to a standard shortest path.
\end{proof}

\section{Non-zero initial charges}\label{S-initial-charges}

Up to this point, we considered optimal travel plans from~$s$ to~$t$ in which the car starts from~$s$ with an empty battery. (In particular, if no recharging is possible at~$s$ and all outgoing arcs of~$s$ have positive costs, the car cannot move out of~$s$.)  It is natural to also consider the case in which the car starts at~$s$ with some fixed initial charge~$a$. (For example, $a=B$ corresponds to the case in which the car starts from~$s$ with a full battery.) 

\subsection{Fixed initial charges}

A simple reduction allows us to solve, within the same asymptotic time bounds, the all-pairs version of the \MCP\ problem in which each vertex $v\in V$ has a fixed initial charge $0\le a(v)\le B$ associated with it. The initial charges at different vertices are not necessarily the same.

For each original vertex $v\in V$ create a new vertex $v_0$ and add an arc $(v_0,v)$ whose cost is $-a(v)$. Let $r(v_0)=\infty$, i.e., no recharging is possible at~$v_0$. (Alternatively, we could let $r(v_0)=r(v)$.) Starting at~$v_0$ with an empty battery corresponds to starting at~$v$ with a charge of~$a(v)$. The new graph has $2n$ vertices and $m+n$ vertices. Thus, the asymptotic running times of the algorithms given in Section~\ref{S-algs} are unaffected.

% A simple modification to the problem statement is to allow non-zero initial charge at the vertices. More formally, we are given an initial charge $b_v$ for every $v\in V$ and are asked to compute the all pairs optimal travel plans with the given initial charges. This problem easily reduces to the problem studied in this paper by adding dummy vertices $v'$ for every $v\in V$ and adding arcs $(v',v)$ of cost $c(v',v)=-b_v$. The optimal travel plan from $v$ (starting with $a_v$ charge) is the optimal travel plan from $v'$ in the modified graph (starting with an empty battery). 

% \subsection{Arbitrary initial charge at a given source vertex}
\subsection{Adding an initial charge or a source vertex}

A further natural generalization is to allow an arbitrary initial charge~$a$, not known in advance, at a specified source vertex~$s$. We show that if there are no negative cycles in the graph, then after the all-pairs version of the \MCP\ problem with zero initial charges is solved, we can solve the single-source version of the \MCP\ problem from a given source vertex~$s$ and an initial charge~$a$ in~$O(n^2)$ time.

The problem above corresponds to adding a new vertex~$s_0$ with $r(s_0)=\infty$, i.e., no charging at~$s_0$, and an arc $(s_0,s)$ of cost $-a$ and solving the single-source version of the \MCP\ problem from~$s_0$ with an initial charge of~$0$. We actually solve the more general problem in which we add new vertex~$s_0$ and a collection of arcs $(s_0,v_1),(s_0,v_2),\ldots,(s_0,v_k)$ to an arbitrary number of original vertices of the graph, each with possibly different cost $c(s_0,v_i)=-a_i\le 0$. This corresponds to the option of using one of several available electric cars, where the $i$-th available car is at vertex~$v_i$ and has an initial charge of~$a_i$.

We present a reduction from the problem of adding a new source vertex to a graph on which the all-pair \MCP\ problem was already solved, to the solution of two single-source \MFC\ problems and two vector-matrix min-plus products. When there are no negative cycles in the graph, the two single-source \MFC\ problems can be solved in $O(m+n\log n)$ time, using a feasible potential function computed while solving the all-pairs \MFC\ and \MIC\ problems on the original graph. (See \cite{DKTZ23} for details.) The two $1\times 2n$ by $2n\times 2n$ min-plus products can be computed na\"{i}vely in $O(n^2)$ time.

% When there are no negative cycles in the graph, the \MFC\ and \MIC\ algorithms of Dorfman et al. \cite{DKTZ23} start by computing a feasible potential function. This potential function can be extended into a potential function for the new graph. Using this potential function it is easy to compute $\alpha_{B,0}(s_0,v), \alpha_{B,B}(s_0,v)$ for every $v\in V$ in $O(m+n\log n)$ time. (See \cite{DKTZ23} for details.)

We start by computing $\alpha_{B,0}(s_0,v), \alpha_{B,B}(s_0,v)$ for every $v\in V$. These are the two instances of the single-source \MFC\ problem.

Next, as in Section~\ref{S-reduction}, we compute $\ell_B(s_0^0,v^b)$, for every $v\in V$ and $b\in\{0,B\}$. Recall that $\ell_B(s_0^0,v^b)$ is the minimum cost of getting from~$s_0$ to $v$, starting with an empty battery and arriving with either an empty or full battery and using at most one recharging. This can be done in $O(n^2)$ time by na\"{i}vely computing a vector-matrix min-plus product.

Now, since we already know $\delta_{\ell_B}(u^a,v^0)$, for every $u,v\in V$ and $a\in\{0,B\}$, we can compute $\rho_B(s,v)$, for every $v\in V$ using an additional vector-matrix min-plus product, again in~$O(n^2)$ time. More concretely, for every $v\in V$, we have $\rho_B(s_0,v)=\min_{u\in V,a\in\{0,B\}} \ell_B(s_0^0,u^a)+\delta_{\ell_B}(u^a,v^0)$. Correctness follows again by Lemma~\ref{L-y}.

\section{Concluding remarks}\label{S-concl}

We have presented a simple reduction from the problem of computing Minimum-Cost Plans (\MCP\,) between all pairs of vertices in a graph to two simpler problems: The first is the computation of optimal energetic paths when no recharging of the battery is allowed. This problem comes in two equivalent variants: Maximum Final Charge (\MFC\,) and Minimum Initial Charge (\MIC\,). The second problem is the standard All-Pairs Shortest Paths problem.

Using this simple reduction we have obtained an $O(\frac{n^3}{2^{c\sqrt{\log n}}}+mn)$-time algorithm for solving the \MCP\ in graphs with no negative cycles. 
%The running time of the algorithm is always at most $O(n^3)$. 
This matches the running time of the fastest APSP algorithm of Williams \cite{Williams18,Williams21}, unless~$m$ is extremely close to~$n^2$.

An interesting open problem is whether there is an $\tilde{O}(mn)$-time algorithm for the all-pairs MCP problem in $n$-vertex, $m$-arc graphs, essentially matching the complexity of the standard APSP problem  for sparse graphs.

Another interesting open problem is whether there is an $\tilde{O}(n^3)$-time algorithm for the all-pairs MCP problem when the input graph may contain negative cycles. (An interesting feature of the MCP problem is that minimum-cost plans are well-defined, and are of finite length, even in this case. They may not be simple, however.)

All algorithms presented assume that $B$, the maximum capacity of the battery, and $\Delta$, the maximum number of rechargings allowed on each path, are fixed and known in advance. Is it possible to obtain an efficient algorithm that preprocesses an input graph and is then able to quickly answer queries of the form $\rho_{B,\Delta}(s,t)$, i.e., what is the minimum cost of a plan from~$s$ to~$t$ when the capacity of the battery is~$B$ and at most~$\Delta$ rechargings are allowed on the way from~$s$ to~$t$?

Finally, we remark that we considered the natural problem of finding minimum-cost plans, ignoring the time it takes to traverse arcs or to charge the battery. Finding a  plan that can be implemented within a given time limit is easily seen to be a NP-hard problem by a reduction from $0$-$1$ knapsack. Similarly, the problem becomes NP-hard if each arc $(u,v)$ has a monetary cost $r(u,v)$, i.e., a \emph{toll} that should be paid to traverse the arc, in addition to its energetic cost~$c(u,v)$, and the goal is to minimize the cost of travelling from~$s$ to~$t$. It may be possible to obtain interesting approximation algorithms, but this is beyond the scope of the current paper.

\section*{Acknowledgement}
We would like to thank two anonymous SOSA reviewers for valuable comments that helped improve the paper.

\bibliographystyle{plain}
\bibliography{main}

\end{document}